\documentclass[11pt, oneside]{article}   	
\usepackage[english]{babel}
\usepackage[toc,page]{appendix}
\usepackage[margin=1in]{geometry}  
\usepackage{hyperref}
\usepackage{graphicx,wrapfig}
\usepackage[labelformat=empty]{caption}
\usepackage{yfonts}	

\usepackage[lined,boxed,ruled,norelsize,algo2e,linesnumbered]{algorithm2e}
\usepackage{thmtools}
\usepackage{xspace}
\newcommand{\sse}{\subseteq}

\usepackage{amssymb}
\usepackage{amsmath}
\usepackage{amsthm}
\usepackage{todonotes}
\newcommand{\vol}{\mathrm{vol}}
\newcommand{\cg}{\mathrm{cg}}
\newcommand{\Var}{\mathrm{Var}}
\newcommand{\E}{\mathbb{E}}

\newcommand{\R}{\mathbb{R}}

\newcommand{\cbc}{convex body chasing\xspace}

\newtheorem{thm}{Theorem}[section]
\newtheorem{lem}[thm]{Lemma}
\newtheorem{prop}[thm]{Proposition}
\newtheorem{cor}[thm]{Corollary}

\theoremstyle{definition}
\newtheorem{defn}{Definition}
\newtheorem*{rem}{Remark}

\title{Chasing Nested Convex Bodies Nearly Optimally}
\author{S\'ebastien Bubeck\\
       Microsoft Research
       \and
Bo'az Klartag \\
Weizmann Institute
       \and
       Yin Tat Lee\\
       University of Washington \\ \& Microsoft Research
       \and
      Yuanzhi Li\thanks{This work was done while M. Sellke and Y. Li were at Microsoft Research.}\\
       Stanford University
      \and
      Mark Sellke\footnotemark[1]\\
       Stanford University\\
}
\date{}							

\begin{document}

\maketitle

\begin{abstract}

The convex body chasing problem, introduced by Friedman and Linial \cite{FriedmanLinial}, is a competitive analysis problem on any normed vector space. In convex body chasing, for each timestep $t\in\mathbb N$, a convex body $K_t\subseteq \mathbb R^d$ is given as a request, and the player picks a point $x_t\in K_t$. The player aims to ensure that the total distance moved $\sum_{t=0}^{T-1}||x_t-x_{t+1}||$ is within a bounded ratio of the smallest possible offline solution.

In this work, we consider the nested version of the problem, in which the sequence $(K_t)$ must be decreasing. For Euclidean spaces, we consider a memoryless algorithm which moves to the so-called Steiner point, and show that in an appropriate sense it is exactly optimal among memoryless algorithms. For general finite dimensional normed spaces, we combine the Steiner point and our recent algorithm in \cite{ABCGL} to obtain a new algorithm which is nearly optimal for all $\ell^p_d$ spaces with $p\geq 1$, closing a polynomial gap.
\end{abstract}
\thispagestyle{empty}

\setcounter{page}{1}
\section{Introduction}

We study a version of the \emph{convex body chasing} problem. In this problem, a sequence of $T$ requests $K_1,\dots,K_T$, each a convex body in a $d$-dimensional normed space $\mathbb R^d$, is given. The algorithm starts at a given position $x_0$, and each round it sees the request $K_t$ and must give a point $x_t\in K_t$. The online algorithm aims to minimize the total movement cost 
\[\sum_{t=0}^{T-1} ||x_{t+1}-x_t||.\]
More precisely, we take the viewpoint of competitive analysis, in which the algorithm aims to minimize the competitive ratio of its cost and the optimal in-hindsight sequence $y_t\in K_t$. 

Our companion paper \cite{BLLS2018} establishes the first finite upper bound for the competitive ratio of convex body chasing. Our upper bound is exponential in the dimension $d$, while the best known lower bound is $\Omega(\sqrt{d})$ in Euclidean space, which comes from chasing faces of a hypercube.

In this paper we consider a restricted variant of the problem, \emph{nested convex body chasing}, in which the bodies must be decreasing:
\[K_1\supseteq K_2\supseteq \dots \supseteq K_T.\]

This problem was first considered as a potentially more tractable version of the full problem in \cite{bansal2018nested}, where it was shown to have a finite competitive ratio. In our previous work \cite{ABCGL} we gave an algorithm with nearly linear $O(d\log d)$ competitive ratio for any normed space. This upper bound is nearly optimal for $\ell^{\infty}$ but for most normed spaces there is still a gap. For example for $\ell^2$, the best lower bound is $\Omega(\sqrt{d})$. 

\subsection{Notation and a Reduction}

We denote by $B_1$ a unit ball $B(0,1)$ centered at the origin in $\mathbb R^d$. We denote by $\mathcal K$ the space of convex bodies in $\mathbb R^d$. We often refer to the \emph{Hausdorff metric} on $\mathcal K$, defined by 
\[
d_{\mathcal H}(K,K')\equiv\max_{k\in K,k'\in K'}\max(d(k,K'),d(k',K)).
\]
In other words, it is the maximum distance from a point in one of the sets to the other set. Note that for $K_1\supseteq K_2$, we have $d_{\mathcal H}(K_1,K_2)\leq 1$ iff $K_2+B_1\supseteq K_1$.

We denote by $s(K)$ the Steiner point of $K$, defined at the start of section~\ref{sec:steinerO(d)}. We denote by $cg(K)$ the centroid of $K$. 

For nested chasing, it turns out we can simplify the problem to make it significantly easier to think about. The following reduction is taken from \cite{ABCGL}, except that we add in a dependence on the number $T$ of requests. The proof is essentially unchanged.

\begin{lem}[\cite{ABCGL}] \label{lem:reduction}
For some function $f(d,T) \geq 1$, the following three propositions are equivalent:
  \begin{itemize} \itemsep 1pt
  \item[(i)] (General) There exists a $f(d,T)$-competitive algorithm for nested \cbc.
  \item[(ii)] (Bounded) Assuming that
  $K_1 \sse B_1$ and $x_0=0$, there exists an algorithm for nested \cbc with total movement 
  $O(f(d,T))$. 
\item[(iii)] (Tightening) Assuming that
  $K_1 \sse B_1$ and $x_0=0$, there exists an algorithm for nested \cbc that
  incurs total movement cost $O(f(d,T))$ until the first time $t$
  at which 
  $K_t$ is
  contained in some ball of radius $\frac{1}{2}$.
  \end{itemize}
\end{lem}

\subsection{Our Results}

This paper has two essentially independent parts. In the first part, we consider the algorithm which moves to the newly requested body's \emph{Steiner point}, a geometrical center which is defined for convex subsets of Euclidean spaces. We show that for the Bounded version of nested chasing, it achieves a competitive ratio $O(\min(d,\sqrt{d \log T}))$ which is nearly optimal for sub-exponentially many requests. We also consider the problem of finding the best memoryless algorithm, meaning that $x_t\in K_t$ must be a deterministic function of only $K_t$. In this case, we compare the competitive ratio against the Hausdorff distance between $K_1$ and $K_T$, or equivalently the in-hindsight optimum starting from the worst possible $x_1$. We show that the Steiner point achieves the \emph{exact} optimal competitive ratio for any $(d,T)$ by adapting an argument from \cite{steinerlipschitz}.

In the second part of this paper, we give a different algorithm which is nearly optimal even for exponentially many requests. Our previous algorithm from \cite{ABCGL} achieved a $O(d\log d)$ competitive ratio by moving to the centroid and recursing on short dimensions; like that algorithm, our new algorithm is most naturally viewed in terms of Tightening. Inspired by Steiner point, we improve that algorithm by adding a small ball $B_r$ to $K_t$ and taking the centroid with respect to a log-concave measure which depends on the normed space. Using this procedure, we obtain a general algorithm for any normed space. For Euclidean spaces, our new algorithm is  $O(\sqrt{d\log d})$ competitive. For $\ell^p$ spaces with $p\geq 1$, our new algorithm is optimal up to a $O(\log d)$ factor.







\section{The Steiner Point and Competitive Ratio $d$}
\label{sec:steinerO(d)}
Here we define the Steiner point and explain why it achieves competitive ratio $d$ (in the \emph{Bounded} formulation (ii) of Lemma~\ref{lem:reduction}).

\begin{defn}

For a convex body $K\subset \R^d$, its \emph{Steiner point} $s(K)$ is defined in the following two equivalent ways:

\begin{enumerate}
    \item For any direction $\theta\in \mathbb{S}^{d-1}\subset \R^d$, let $f_K(\theta)=\arg\max_{x\in K}(\theta\cdot x) $ be the extremal point in $K$ in direction $\theta$. Then compute the average of this extremal point for a random direction: \[s(K)=\int_{\theta\in \mathbb{S}^{d-1}} f_K(\theta) d\theta. \]
    
    Here we integrate over the uniform (isometry invariant) probability measure on $\mathbb{S}^{d-1}$.
    
    \item For any direction $\theta\in \mathbb{S}^{d-1}\subset \R^d$, let $h_K(\theta)=\max_{x\in K}(\theta\cdot x)$ be the support function for $K$ in direction $\theta$, and compute
    \[s(K)=d\int_{\theta\in \mathbb{S}^{d-1}} h_K(\theta)\theta d\theta. \]
\end{enumerate}

\end{defn}

The equivalence of the two definitions follows from $\nabla h_{K}(\theta)=f_{K}(\theta)$ and the divergence theorem. The first definition makes it evident that $s(K)\in K$. We will use the second definition to control the movement. Behold:

\begin{thm}
\label{thm:steinerO(d)}

Let $B_1 = K_1 \supseteq K_2\supseteq K_3\supseteq\dots \supseteq K_T$ be a sequence of nested convex bodies. Then
\[\sum_{t=1}^{T-1}||s(K_t)-s(K_{t+1})||_2\leq d.
\]
More generally, for any sequence of nested convex bodies \[K_1\supseteq K_2\supseteq\dots\supseteq K_T\] we have the estimate
\[\sum_{t=1}^{T-1}||s(K_t)-s(K_{t+1})||_2\leq \frac{d}{2}\cdot(w(K_1)-w(K_T)).\]
where $\omega(\cdot)$ denotes the mean width, the average length of a random $1$-dimensional projection.

\end{thm}

\begin{proof}

The idea is simply that for each fixed $\theta$, the integrand decreases by a total of at most $2$ over all the requests, so the total budget for movement is $2d$. To save the factor $2$ we combine $\pm \theta$, noting that they can change by at most $2$ in total. Writing it out:
\[
\sum_{t=1}^{T-1}||s(K_t)-s(K_{t+1})||_2\leq d\int_{\theta\in \mathbb{S}^{d-1}} \sum_{t=1}^{T-1} |h_{K_t}(\theta)-h_{K_{t+1}}(\theta)|d\theta. 
\]
Now, as $K_{t}\supseteq K_{t+1}$ we have $h_{K_t}(\theta)\geq h_{K_{t+1}}(\theta)$ and so 
\begin{align*}
\int_{\theta\in \mathbb{S}^{d-1}} \sum_{t=1}^{T-1} |h_{K_t}(\theta)-h_{K_{t+1}}(\theta)|d\theta
& = \int_{\theta\in \mathbb{S}^{d-1}} \sum_{t=1}^{T-1} (h_{K_t}(\theta)-h_{K_{t+1}}(\theta)) d\theta \\
& \leq \left(\int_{\theta\in \mathbb{S}^{d-1}} h_{K_1}(\theta)d\theta\right)-\left(\int_{\theta\in \mathbb{S}^{d-1}} h_{K_T}(\theta)d\theta\right).
\end{align*}

The first integral is at most $1$ because $K_T \subseteq B_1$. As $h_K(\theta)+h_K(-\theta)\geq 0$ for any convex body $K$, the second integral is non-negative. So we conclude that 
\[\sum_{t=1}^{T-1}||s(K_t)-s(K_{t+1})||_2\leq d\]
which proves the theorem. The proof of the more general statement is identical.
\end{proof}

\section{Better Competitive Ratio for Subexponentially Many Queries}

Here we refine the argument from the previous section to show that with $T$ rounds, we obtain a competitive ratio $O(\sqrt{d\log T})$. Hence, for subexponentially many queries, the Steiner point is a nearly optimal algorithm.

The intuition for this improved estimate is the following. We think of $h_K$ as a budgeted resource with amount $d\int_{\theta}h_{K_t}(\theta)-h_{K_{T}}(\theta) d\theta$ at time $t$. Our goal is to turn this budget into Steiner point movement. In order to have a lot of movement after a modest number of time-steps, we must use a significant amount of this budget in a typical time-step. So, suppose that a lot of the budget is used in going from $K_t\to K_{t+1}$. Then because only a tiny subset of the sphere correlates significantly with any fixed direction, the values of $\theta$ contributing to the movement \[s(K_t)-s(K_{t+1})=d\int_{\theta\in \mathbb{S}^{d-1}}(h_{K_t}(\theta)-h_{K_{t+1}}(\theta))\theta d\theta\] must include a wide range of directions $\theta$. As a result, the different $\theta$ values contribute to the Steiner point movement in very different directions and cancel out a lot, which means that the Steiner point's movement was actually much less than the amount of budget used up. In other words, the starting budget $d\cdot h_{K_0}(\theta)$ can only be used efficiently if it is used very slowly. Below we make this argument precise using concentration of measure on the sphere.

\begin{lem}[{Concentration of Measure \cite[Lemma 2.2]{ball1997elementary}}]
For any $0\leq \varepsilon <1$, the cap $$\{x\in \mathbb{S}^{d-1}:x_{i}\geq \varepsilon\}$$ has
measure at most $e^{-d \varepsilon^{2}/2}$.
\end{lem}

\begin{lem}\label{lem:bound_s_refine}
For any convex bodies $K'+B_1\supseteq K\supseteq K'$, setting $\lambda=\frac{w(K)-w(K')}{2}$
\[||s(K)-s(K')||_2\leq \lambda \sqrt{d \log(\lambda^{-1})}. \]
\end{lem}

\begin{proof}
We have 
\[ s(K)-s(K') =d\int_{\theta}  \left(h_{K}(\theta)-h_{K'}(\theta)\right) \theta d\theta.\]
Since 
\[ \|s(K)-s(K') \|_2=\sup_{v\in \mathbb{S}^{d-1}} v^{\top}  (s(K)-s(K')),\]
it suffices to estimate $v^\top(s(K)-s(K'))$ for arbitrary unit $v$. Without loss of generality we can just take $v=e_1:=(1,0,\dots,0).$
And now we have
\[v^{\top}(s(K)-s(K'))=d\int_{\theta}  \left(h_K(\theta)-h_{K'}(\theta)\right)  \theta_1 d\theta.\]
Because $K'+B_1\supseteq K\supseteq K'$, we have
\[0\leq h_{K}(\theta)-h_{K'}(\theta)\leq 1.\]
Given this constraint that $h_K-h_{K'}\in [0,1]$ and has total mass $\lambda$, the maximum possible value of the integral is easily seen to be achieved when the $\lambda$-fraction of $\theta$ values with largest possible $\theta_1$ coordinate have $h_K(\theta)-h_{K'}(\theta)=1$, and the rest have $h_K(\theta)-h_{K'}(\theta)=0$. 

Therefore we are reduced to bounding the largest average $\theta_1$ coordinate over a subset of size $\lambda$. Concentration on the sphere shows that $1- e^{- t^2/2}$ fraction of the sphere lies in the set $\{\theta_1 \leq t/\sqrt{d}\}$. Therefore, the average of $\theta_1$ over a subset of size $\lambda$ is at most $O\left(\sqrt{\frac{\log(\lambda^{-1})}{d}}\right)$.
\end{proof}

\begin{thm}
Following the Steiner point, starting from $B_1$, gives total movement 
\[O(\min(d,\sqrt{d\log T}))\]
after $T$ requests. More generally, the same upper bound holds for any sequence of nested convex bodies \[K_1\supseteq K_2\supseteq\dots\supseteq K_T\] with $K_T+B_1\supseteq K_1$. 
\end{thm}

\begin{rem}

The condition $K_T+B_1\supseteq K_1$ means that OPT will always be $1$ from any starting point. Hence this is a statement about the competitive ratio.

\end{rem}

\begin{proof}
We have established an upper bound of $O(d)$ so we prove only the $O(\sqrt{d\log T})$ estimate. Say the requests are 
\[B_1=K_1\supseteq K_2\supseteq \dots\supseteq K_T, \]
with $K_T$ being a singleton.
Setting $\lambda_t=\frac{w(K_t)-w(K_{t+1})}{2}$,
Lemma~\ref{lem:bound_s_refine} shows that 
\[\sum_{t=1}^{T-1} |s(K_t)-s(K_{t+1})|\leq O(d^{1/2})\sum_{t=1}^{T-1}\left(\lambda_t\sqrt{\log(\lambda_t^{-1})}\right).\] 
By concavity of $ \sqrt{\log(\lambda^{-1})}$ on $\lambda\in [0,1]$, and the constraints $\lambda_t \geq 0$ and $\sum_{t=1}^{T-1} \lambda_t = 1$, we have the upper bound 
\[\sum_{t=1}^{T-1}\left(\lambda_t\sqrt{\log(\lambda_t^{-1})}\right)\leq \sqrt{\log T}.\] Therefore
\[\sum_{t=1}^{T-1} |s(K_t)-s(K_{t+1})|=O(\sqrt{d\log T}).\]The proof of the more general statement is identical.
\end{proof}

It is natural to wonder whether exponentially large $T$ actually results in $d$ movement. The next theorem establishes that this is indeed the case. We defer the proof to the appendix.

\begin{restatable}{thm}{steinerlinear}

\label{thm:steinerlinear}

For $T\geq 100^d$, there is a sequence of $T$ convex bodies starting from a unit ball which results in $\Omega(d)$ movement of the Steiner point.

\end{restatable}

\section{Optimal Memoryless Chasing}

In the chasing nested convex bodies problem, the optimal strategy potentially depends on the new request $K_t$ and the initial/current points $x_0$ and $x_{t-1}$. However, our Steiner point algorithm achieved a good guarantee while using only $K_t$ to choose $x_t$. It is natural to ask how well such a \emph{memoryless} strategy with $x_t$ a function of only $K_t$ can do. Hence in this section we restrict our attention to such algorithms, or equivalently to selector functions, and formulate a precise question. We then show the Steiner point is the exact optimal solution to this question.

\begin{defn}
A \emph{selector} is a function $f:\mathcal K\to \mathbb R^d$ such that $f(K)\in K$ for all $K$.
\end{defn}

We formulate the \emph{memoryless} nested convex body chasing problem as follows. We aim to define a selector $f$ which keeps the movement cost $\sum_{t=1}^{T-1}||f(K_t)-f(K_{t+1})||_2$ within a small constant factor of the Hausdorff distance $d_{\mathcal H}(K_1,K_T)$ between $K_1,K_T$, for all sequences of $T$ convex bodies. Note that now we do not begin at a given $x_0$ point, and are instead free to choose the starting point at no cost.

The theorem below shows that with the above formulation, the Steiner point achieves the exact optimum competitive ratio. The result and proof are inspired by a similar result from the work \cite{steinerlipschitz} which proves that the Steiner point achieves the exact optimum Lipschitz constant among all selectors, where $\mathcal K$ is metrized by the Hausdorff metric. This is similar to the $T=2$ case of our problem, and as we remark below, our proof specializes to give their result as well; the nested condition is not crucial.

\begin{restatable}{thm}{steineroptimal}

\label{thm:steineroptimal}

For any $d$ and $T$, the Steiner point achieves the exact optimum competitive ratio for the memoryless nested convex body chasing problem. That is, among all selectors $f$, the Steiner point yields the minimum constant $C(d,T)$ such that the following holds: for any sequence
\[K_1\supseteq K_2\supseteq \dots \supseteq K_T\]
of nested bodies with $K_T+B_1\supseteq K_1$, the total movement is
\[\sum_{t=1}^{T-1} ||f(K_t)-f(K_{t+1})||_2\leq C(d,T). \]
\end{restatable}

\begin{proof}[Description of Proof]

The proof idea is as following: given a selector $f$, we symmetrize it to obtain a new function $\tilde f$ with at most the competitive ratio of $f$, which has some new symmetry. More precisely, $\tilde f$ is equivariant under the isometry group of $\mathbb R^d$ and commutes with addition, where addition of convex sets is Minkowski sum as usual. These symmetry properties are strong enough to force $\tilde f$ to \emph{coincide} with the Steiner point. Since $\tilde f=s$ has a smaller Lipschitz constant than $f$ by construction, and $f$ was arbitrary, the result follows. 

The way we symmetrize $f$ is to take 
\[``\tilde f(K)=\mathbb E[g^{-1}(f(g(K)+K')-f(K'))]"\]
where $g$ is a random isometry of $\mathbb R^d$ and $K'$ is a random convex set. We require that the probability measures for $g,K'$ be invariant under composition with isometries and Minkowski addition of a fixed convex set, respectively; these invariance properties ensure that $\tilde f$ is isometry invariant and additive. The issue is that such probability measures do not actually exist. But by using the concept of an \emph{invariant mean} instead of a probability measure, we get a translation-invariant average that does the same job, though it cannot be written down without the axiom of choice. See the appendix for the precise argument and some references.
\end{proof}

\begin{rem}

The theorem above would still hold with the same proof if the sequence $(K_1,K_2,\dots,K_T)$ were not required to be nested but were simply constrained to satisfy, for all $\theta$, 
\[\sum_{t=1}^{T-1}|h_{K_t}(\theta)-h_{K_{t+1}}(\theta)| \leq 1.\]
For this generalization, the $T=2$ case is exactly the aforementioned result from \cite{steinerlipschitz} that the Steiner point attains the minimum Lipschitz constant among all selectors. This optimal Lipschitz constant is asymptotically  $\left(\sqrt{\frac{2d}{\pi}}\right)$; the reason is similar to the proof of Theorem~\ref{thm:steinerO(d)}.

\end{rem}

\section{Time-independent Nearly Optimal Algorithm with Mirror Map}

In this section, we give a new algorithm which combines the Steiner point and our method from \cite{ABCGL}. This algorithm is nearly optimal for any $\ell^p$ norm with $p\geq 1$, and we conjecture the result to be nearly tight in a general normed space.


Our first observation is inspired by an alternative formula for the Steiner point from e.g. \cite{steinerlimit}:
$$s(K) = \lim_{r \rightarrow + \infty} \cg(K + B_r)$$
where $\cg$ denotes the centroid and $B_r$ is a $\ell^2$ ball centered at $0$ with radius $r$. 

Since both Steiner point and centroid (with some modification proposed in $\cite{ABCGL}$) give a competitive algorithm for nested convex body chasing, it is natural to conjecture that $\cg(K + B_r)$ works for any $r\geq0$ (maybe with some modification). In this section, we show that it is indeed true for a certain range of $r$. The benefit of this variant of the Steiner point (or centroid) is that it is easier to modify the definition to fit our needs for other normed spaces. 

The second observation is that the centroid of a Gaussian measure restricted to any convex set moves with $\ell^2$ distance proportional to its standard deviation when the measure is cut through its centroid. This is due to some concentration properties of Gaussian measure. Therefore, for $\ell^2$ nested chasing bodies, it is more natural to use Gaussian measure instead of uniform measure to define the centroid. 

For a general normed space, we assume there is a $1$-strongly convex\footnote{A function is $\alpha$-strongly convex on a normed space $\|\cdot \|$ if
$\phi(y)\geq\phi(x)+\left\langle \nabla\phi(x),y-x\right\rangle +\frac{\alpha}{2}\|y-x\|^{2}$ for all $x,y$.} function $\phi$ on the space such that $0 \leq \phi(x) \leq D$ for all $\|x \| \leq 1$.
We will use $\phi$ to construct an algorithm with competitive ratio $O(\sqrt{d D \cdot \log{d}})$.
For $\ell^2$ space, we can use $\phi(x) = \|x\|^2_2$ to get competitive ratio $O(\sqrt{d\cdot \log{d}})$. In general, the minimum $D$ among all $1$-strongly convex functions ranges from $\frac{1}{2}$ to $\frac{d}{2}$. The $\frac{1}{2}$ lower bound follows from applying the strong convexity to $(x,0,-x)$ for $|x|=1$, while a linear upper bound follows from John's theorem. This minimum $D$ measures the complexity of the normed space for online learning \cite{srebro2011universality}.

Now, we define the weighted centroid as follows:
\begin{defn}
We define the centroid and the volume of $K$  with respect to $e^{-\phi(x)}dx$ by
\[\cg_{\phi}(K)=\frac{\int_K e^{-\phi(x)}xdx}{\int_K e^{-\phi(x)}dx}
\qquad \text{ and } \qquad
\vol_{\phi}(K) = \int_K e^{-\phi(x)}dx.\]
\end{defn}

Algorithm \ref{alg:generalnorm} is based on $\cg_{\phi}(K + B_r)$ and we measure the progress by $\vol_{\phi}(K + B_r)$. The key lemma we show is that this mixture of Steiner point and centroid is stable under cutting and that the volume $\vol_{\phi}(K + B_r)$ decreases by a constant factor every iteration. Unlike the standard volume $\vol_{\phi}(K)$, this volume has a lower bound, which is related to the volume of the ball $B_r$. Therefore, the algorithm terminates without the trick of projecting the body from $\cite{ABCGL}$. We note that the projection trick in that paper does not work because the total movement of $\cg$ during the projection is already $\Omega(d)$.

\begin{restatable}{thm}{movementcgr}
\label{thm:movement_cg_r}Let $\phi$ be an $\alpha$-strongly convex
function on the normed space $\|\cdot\|$ and scalar $r>0$. Let $K$
be a convex set and $v$ be an unit vector. Let $H$ be the half space
with normal $v$ through $\cg_{\phi}(K+B_{r})$:
\[
H=\{x\in\R^{d}:v^{\top}x\geq v^{\top}\cg_{\phi}(K+B_{r})\}.
\]
Suppose that $\Var_x v^{\top}x\geq(2e\cdot r)^{2}$ where $x$ is sampled from $e^{-\phi(y)}1_{y\in K+B_r}/\vol_\phi(K+B_r)$. Then
\[
\frac{1}{e}\cdot\vol_{\phi}(K+B_{r})\leq\vol_{\phi}((K\cap H)+B_{r})\leq(1-\frac{1}{2e})\cdot\vol_{\phi}(K+B_{r})
\]
and that $\|\cg_{\phi}((K\cap H)+B_{r})-\cg_{\phi}(K+B_{r})\|\lesssim\alpha^{-\frac{1}{2}}.$
\end{restatable}


Using this geometric statement, one can readily prove the following statement by choosing appropriate parameters. 
%
%
%


\begin{figure}
\SetKwFor{Loop}{loop}{}{end}
\begin{algorithm2e}[H]\label{alg:generalnorm}

\caption{$\mathtt{ChasingNormedSpace}$}

\SetAlgoLined\DontPrintSemicolon

\textbf{input: }a normed space $\|\cdot\|$ such that $\|x\|_{2}\leq\|x\|\leq\sqrt{d}\cdot\|x\|_{2}$
for all $x$,

\textbf{\hphantom{\textbf{input:}}} $\phi$ is a $\alpha$-strongly
convex function on $\|\cdot\|$ such that $0\leq\phi(x)\leq D$
for all $\|x\|\leq2$,

\textbf{\hphantom{\textbf{input:}}} a $C$-competitive $\ell^{2}$
nested chasing convex algorithm (used for narrow directions),

\textbf{\hphantom{\textbf{input:}}} padded radius $r\leq\frac{1}{\sqrt{d}}$.

Let the localized set $\Omega=\{x:\|x\|\leq1\}$ and $K$ be the current
convex set.

Take the initial point $x=\cg_{\phi}(\Omega+B_{r})$ where $\cg_{\phi}(K):=\frac{\int_{K}x\cdot e^{-\phi(x)}dx}{\int_{K}e^{-\phi(x)}dx}$ and  $B_{r}$ is a $\ell^{2}$ ball centered at $0$ with radius $r$.

\Loop{}{

Let $A$ be the covariance matrix of the distribution $e^{-\phi(y)}1_{y\in\Omega+B_{r}}$.

Let $V$ be the span of eigenvectors of $A$ with eigenvalues less
than $(2e\cdot r)^2$.

\uIf{$V=\{0\}$}{

Receive a new convex set $K$ such that $x\notin K$.

}\Else{

Run the nested chasing algorithm on $\Omega\cap(x+V)$ until the new
$K$ satisfies $(x+V)\cap K=\emptyset$.

}

Let $H$ be a half space containing the convex set $K$ such that
$x+V\subset\partial H$, namely, a half space touching $x$ and parallel to $V$.

$\Omega\leftarrow\Omega\cap H$.

$x\leftarrow\cg_{\phi}(\Omega+B_{r})$.\tcp*[f]{If $x\notin \Omega$, pay extra $2r$ movement to $\Omega$ and back.}

}
\end{algorithm2e}
\end{figure}














\begin{restatable}{thm}{normedchasing}
\label{thm:normedchasing}
For any normed space $\|\cdot\|$ on $\R^{d}$ equipped with a $1$-strongly
convex $\phi$ such that $0\leq\phi(x)\leq D$ for all $\|x\|\leq1$,
there is an $O(\sqrt{dD\cdot\log d})$-competitive nested chasing
convex body algorithm.
\end{restatable}

For $p\geq2$, the function $\phi(x)=\frac{1}{2}\|x\|_{2}^{2}$ is
1-strongly convex in the $\ell^{p}$ norm with $D\leq\frac{1}{2}n^{1-\frac{2}{p}}$.
For $p\in(1,2]$, the function $\phi(x)=\frac{1}{2(p-1)}\|x\|_{p}^{2}$
is 1-strongly convex in the $\ell^{p}$ norm with $D\leq\frac{1}{2(p-1)}$
(See for instance \cite[Lemma 17]{shalev2007online}). For $p< 1+\frac{1}{\log d}$, we can use the same function above
with $p=1+\frac{1}{\log d}$. Hence, we have the following
bound:
\begin{cor}
For $1\leq p\leq\infty$, the competitive ratio for chasing nested convex bodies
in the $\ell^{p}$ norm is at most
\[
\begin{cases}
O(d^{1-\frac{1}{p}}\sqrt{\log d}) & \text{if }p\geq2\\
O\left(\sqrt{\frac{d\log d}{p-1}}\right) & \text{if }1+\frac{1}{\log(d)}\leq p\leq2 \\
O(\sqrt{d}\log d) & \text{if } p < 1+\frac{1}{\log(d)}
\end{cases}.
\]
\end{cor}

Now, we prove that this bound is tight up to a $O(\log d)$ factor.
\begin{lem}
The competitive ratio in $\ell^{p}$ norm
is at least $\Omega(\max(\sqrt{d},d^{1-\frac{1}{p}}))$ for any $p\geq 1$.
\end{lem}

\begin{proof}
In \cite{FriedmanLinial}, Friedman and Linial used the family $K_{t}=K_{t-1}\cap\{x_{t}=\pm1\}$
to conclude that the competitive ratio is at least $\sqrt{d}$ for
the $\ell^{2}$ norm because the offline optimum can directly go to
the singleton $K_{d}$ using $\sqrt{d}$ movement in the $\ell^{2}$ norm
while the online algorithm must move $d$ distance in the $\ell^{2}$
norm. We note that this proof also gives $d^{1-\frac{1}{p}}$ competitive ratio lower
bound for the $\ell^{p}$ norm.

Now, we prove the $\sqrt{d}$ lower bound for any $\ell^{p}$
norm. In this lower bound, we assume without loss of generality that
the dimension $d$ is a power of $2$. Consider the initial point
is $0$ and the initial convex body $K_{0}=\mathbb{R}^{d}$. Let $H$
be the $d\times d$ Hadamard matrix and $h_{t}$ be the $i$-th row
of $H$.

We construct the adaptive adversary sequence as follows: Let $x_{t}\in K_{t}$
be the response of the algorithm. If $h_{t}^{\top}x_{t}\geq0$, we
define
\[
K_{t+1}=K_{t}\cap\{y\in\mathbb{R}^{d}:h_{t}^{\top}y=-1\}
\]
else, we define it by
\[
K_{t+1}=K_{t}\cap\{y\in\mathbb{R}^{d}:h_{t}^{\top}y=+1\}.
\]

Due to the construction, we have that $\left|h_{t}^{\top}(x_{t+1}-x_{t})\right|\geq1$.
Since each entry of $h_{t}$ is $\pm1$, the movement in the $\ell^{p}$
norm is at least
\[
\min_{\left|h_{t}^{\top}\delta\right|\geq1}\|\delta\|_{p}=\min_{\sum_{i=1}^{d}\delta_{i}=1}\|\delta\|_{p}=d^{\frac{1}{p}-1}.
\]
Hence, the algorithm must move $d^{\frac{1}{p}-1}$ in the $\ell^{p}$
norm each step. After $d$ iterations, the algorithm must move $d^{\frac{1}{p}}$
in the $\ell^{p}$ norm. On the other hand, Since $H$ is invertible, $K_{d}$ consists exactly one point $x^{*}$. The offline
optimum can simply move from $0$ to $x^{*}$ at the first iteration.
Note that
\[
Hx^{*}=s
\]
for some $\pm1$ vector $s$. Since the minimum spectral value
of $H$ is exactly $\sqrt{d}$, we have that
\[
\|x^{*}\|_{2}\leq d^{-\frac{1}{2}}\|s\|_{2}\leq1
\]
and hence $\|x^{*}\|_{p}\leq d^{\frac{1}{p}-\frac{1}{2}}$. Therefore,
the competitive ratio is at least $\frac{d^{\frac{1}{p}}}{\|x^{*}\|_{p}}\geq\sqrt{d}$.
\end{proof}






















\subsubsection*{Acknowledgement}
This work was supported in part by NSF Awards CCF-1740551, CCF-1749609, and DMS-1839116. We thank Michael Cohen, Anupam Gupta, Daniel Sleator, and Jonathan Weed for helpful discussions and comments. 

\appendix

\bibliographystyle{alpha}
\bibliography{ChasingNested}

\appendix 

\section{Appendix}

\subsection{Lower Bounds for Steiner Point and Variants}

\subsubsection{Steiner Point Has $\Omega(d)$ Movement with Exponentially Many Requests}

Here we verify that the Steiner point can in fact saturate the upper bound $O(\min(d,\sqrt{d\log T}))$ when $T$ is exponentially large.

\steinerlinear*

\begin{proof}

Pick a $\frac{1}{20}$-net $\{v_1,\dots,v_T\}$ on the sphere. We know that for large $d$, the net has $T\leq 100^d$ points. Define the half-space
\[H_t=\{x\in B_1 | x^\top v_t \leq 0.9\}\]
and take $K_t=K_{t-1}\cap H_t$. We claim that this sequence results in $\Omega(d)$ movement of the Steiner point.

To show that this works, first note that $B_{0.9}\subseteq K_T\subseteq B_{0.95}$. The first inclusion is obvious. To see the second, note that if $1\geq\|x\|_2>0.95$, then for $t$ with $\left|\frac{x}{|x|}-v_t\right|<0.05$ we have $v_t^\top x > \|x\|_2-\frac{1}{20} \geq 0.95$ so $x\notin H_t$.

The fact that $K_T\subseteq B_{0.95}$ means that for all $\theta$, we have $h_{K_T}(\theta)\leq 0.95$ which means the total budget decrease is large:
\[d\int_{\theta} 1-h_{K_m}(\theta) d\theta \geq \frac{d}{20}.\]
On the other hand, we claim that we use our budget at constant efficiency:
\begin{equation}
\label{eq:steinerefficient}
\|s(K_t)-s(K_{t+1})\|_2 \geq \frac{d}{10}\int_{\theta} h_{K_t}(\theta)- h_{K_{t+1}}(\theta)d\theta .
\end{equation}
Once we establish this claimed inequality, the theorem is proved: we used up a constant amount of the budget at constant efficiency, so we achieved $\Omega(d)$ movement. Below, we establish inequality (\ref{eq:steinerefficient}). The point is that since we only cut off a small part each time, we only change the support function $h_{K_t}(\theta)$ for a small-diameter set of directions $\theta$, which cannot have much cancellation. More precisely, we show below that any $\theta_t$ for which $h_{K_t}(\theta_t)\neq h_{K_{t+1}}(\theta_t)$ must satisfy $v_{t+1}^\top \theta_t \geq \frac{1}{10}$, so that all $\theta_t$ contributions point significantly in the direction of $v_{t+1}$. From this, inequality (\ref{eq:steinerefficient}) follows easily.

To see this last claim, we first observe that for any point $y_t\not\in H_{t+1}$ which is removed, since $v_{t+1}^\top  y_t\geq 0.9$ we know $y_t$ is close to $v_{t+1}$:
\[\|y-v_{t+1}\|_2 = \sqrt{\|y\|^2 + \|v_{t+1}\|^2-2v_{t+1}^{\top}y_t} \leq \sqrt{2-2v_{t+1}^{\top}y_t}\leq \frac{1}{\sqrt{5}} \leq 0.45.\]
We also observe that if $\theta_t\in \mathbb{S}^{d-1}$ is such that $h_{K_t}(\theta_t)> h_{K_{t+1}}(\theta_t)$, then taking $y_t=\arg\max_{y\in K_t}(\theta_t\cdot y)$ we must have $y_t\not\in H_{t+1}$ and $y_t^\top\theta_t \geq 0.9$. So for this choice of $y_t$ we also have $\|y_t-\theta_t\|_2\leq 0.45$.

As a result, any such $\theta_t$ must be within $0.45+0.45=0.9$ of $v_{t+1}$. Hence, this $\theta_t$ must satisfy $v_{t+1}^\top\theta_t\geq 0.1$. Since all the affected $\theta_t$ vectors are correlated with a common vector $v_{t+1}$, we get
\begin{align*}
\|s(K_t)-s(K_{t+1})\|_2 & = d\left \|\int_{\theta} \left(h_{K_t}(\theta)- h_{K_{t+1}}(\theta)\right)\theta d\theta \right\|_2 \\
& \geq d \int_{\theta} \left(h_{K_t}(\theta)- h_{K_{t+1}}(\theta)\right)(v_{t+1}^\top\theta_t) d\theta \\
& \geq \frac{d}{10} \int_{\theta} \left(h_{K_t}(\theta)- h_{K_{t+1}}(\theta)\right)d\theta .
\end{align*}
This verifies the claimed inequality (\ref{eq:steinerefficient}).
\end{proof}

\subsubsection{Simple Optimizations to the Steiner Point Algorithm Do Not Help}

It is natural to suggest that the Steiner point algorithm often moves unnecessarily. For example, in the $\Omega(d)$-movement example from Theorem~\ref{thm:steinerlinear}, the original Steiner point is in every single convex body request, so $0$ movement is trivially attainable. Here we consider two natural improvements to the Steiner point algorithm. The first moves to the new Steiner point only when it is forced to move, while the second moves towards the new Steiner point until it reaches the boundary of the newly requested set, and then stops. In both cases, we show that we can turn any hard instance for the ordinary Steiner point algorithm in $d$ dimensions into a hard instance for the modified algorithm in $d+1$ dimensions.

\begin{prop}

Suppose the Steiner point starting from a unit ball in $\mathbb R^d$ has movement $C(d,T)$ movement after some sequence of $T$ requests. Then there is a sequence of $T$ requests in $\mathbb R^{d+1}$ which give $\Omega(C(d,T))$ movement for each of the following two algorithms:

\begin{enumerate}
  \item If $x_t\in K_{t+1}$, do not move. If $x_t\not\in K_{t+1}$, move to the Steiner point: $x_{t+1}=s(K_{t+1})$.
  \item If $x_t\in K_{t+1}$, again do not move. If $x_t\not\in K_{t+1}$, move in a straight line towards $s(K_{t+1})$ until reaching $K_{t+1}$, then stop.
\end{enumerate}

\end{prop}

\begin{proof}

Suppose $B_1=K_0\supseteq K_1\supseteq\dots\supseteq K_T$ is an example which forces the Steiner point to move distance $\Omega(C(d,T))$. Then we go up to $d+1$ dimensions and take, for some small $\varepsilon>0$,
\[\hat K_t=K_t \times [0,\varepsilon^t],\]
Then because of this extra coordinate, the first algorithm moves every round as long as $\varepsilon<\frac{1}{2}$. The movement induced by $\hat K_t$ is at least that of $K_t$, so we still obtain $\Omega(d)$ movement and every move is forced. Also, we only slightly increased the diameter by going up $1$ dimension, so we lose only a constant factor from this. 

For the second algorithm, note that $x_t$ must move to be within $O(\varepsilon)$ of $s(K_t)$. Indeed, because of the fast decay of the last coordinate, we move $(1-O(\varepsilon))$-fraction of the way from $x_{t-1}$ to $s(K_t)$, and the total distance to move is $O(1)$. Because $d(x_t,s(K_t))=O(\varepsilon)$, the change in the total movement cost is $O(T\varepsilon)$. Hence by taking $\varepsilon$ small we achieve essentially the same movement.
\end{proof}

\subsection{Invariant Means and the Proof of Theorem~\ref{thm:steineroptimal}}

\subsubsection{Definition of Invariant Mean}

Here we define invariant means and state the result we need. Some references aside from the article \cite{steinerlipschitz} are the book \cite{edwards1994functional} and the blog post \cite{taopost}. 

\begin{defn}

A locally compact Hausdorff semigroup $G$ is \emph{amenable} if $G$ admits a (right) \emph{invariant mean}, i.e. a linear map
\[\Lambda:L^{\infty}(G)\to\mathbb R\]
such that:
\begin{enumerate}
  \item $\Lambda$ has operator norm $1$, and $f\geq 0$ everywhere implies $\Lambda(f)\geq 0$.
  \item We have $\Lambda(g\circ f)=\Lambda(f)$ for all $g\in G$, where $g\circ f(x)=f(xg)$.
\end{enumerate}
\end{defn}

\begin{thm}[{\cite[Section 3.5]{edwards1994functional}, \cite[Proposition 1.2]{steinerlipschitz}}]
Any compact group or abelian semigroup is amenable.
\end{thm}

For a compact group, $\Lambda$ is just the average with respect to Haar measure. For general abelian semigroups it is more complicated.

\begin{cor}
The following semigroups are amenable:
\begin{enumerate}
\item $(\mathbb R^d,+)$
\item The orthogonal group $O(n)$
\item The semigroup $\mathcal K^d$ of convex bodies in $\mathbb R^d$, with Minkowski sum.
\end{enumerate}

\end{cor}

To denote the translation-invariant averaging operator $\Lambda$ we use integral notation following \cite{steinerlipschitz}. So for instance
\[\int_{K\in\mathcal K}f(K) dK:=\Lambda(f)\]
is an average of $f$ with respect to the invariant mean $\Lambda$ on $\mathcal K$.








\subsubsection{Proof of Theorem~\ref{thm:steineroptimal} }

We use invariant means to symmetrize any function into the Steiner point. We first give the axiomatic characterization due to Rolf Schneider of the Steiner point, which we will use after symmetrizing to show that we ended up with the Steiner point.

\begin{lem}[{\cite{schneider1971steiner}}]
\label{lem:steinerunique}

Let $s:\mathcal K\to \mathbb R^d$ be a function from convex sets to $\mathbb R^d$ such that:
\begin{enumerate}
\item $f(K_1)+f(K_2)=f(K_1+K_2)$
\item $f(gK)=gf(K)$ for any isometry $g:\mathbb R^d\to\mathbb R^d$
\item $f$ is uniformly continuous with respect to the Hausdorff metric.
\end{enumerate}
Then $f(K)=s(K)$ is the Steiner point.
\end{lem}

Now we prove that the Steiner point has the lowest movement among all selectors.

\steineroptimal*

\begin{proof}

Suppose $f:\mathcal K\to\mathbb R^d$ is an arbitrary selector which achieves the above movement estimate with constant $C(d,T)$. We first claim that $f$ has to be $2C(d,T)$-Lipschitz. Indeed, suppose we are given $K,K'$. We show that \[||f(K)-f(K')||_2\leq 2C(d,T) d_{\mathcal H}(K,K').\] Indeed let $K''$ be the convex hull of $K\cup K'$, so that \[h_{K''}(\theta)=\max(h_K(\theta),h_{K'}(\theta)).\] 
Since Hausdorff distance $d_{\mathcal H}$ is the same as $L^{\infty}$ norm of the support function, we know that $d_{\mathcal H}(K,K'')\leq d_{\mathcal H}(K,K')$ and $d_{\mathcal H}(K',K'')\leq d_{\mathcal H}(K,K')$. Since $K''$ contains both $K,K'$, the assumed movement estimate for $f$ tells us that $f(K'')$ is within distance $C(d,T)d_{\mathcal H}(K,K')$ from both $f(K)$ and $f(K')$. Now the claim follows from the triangle inequality.

Now we introduce symmetry to $f$ to obtain the Steiner point. For ease of understanding, we go in three steps. 

First we make $f$ translation invariant:
\[f^{(1)}(A)=\int_{x\in\R^d} f(A+x)-x dx.\]
This is well-defined because $f(A+x)\in A+x$ so the integrand $\|f(A+x)-x\|_2\leq \max_{a\in A}\|a\|_2$ is bounded depending only on $A$.

$f^{(1)}$ is translation invariant because
\begin{align*}
f^{(1)}(A+y) & =\int_{x\in\R^d} f(A+y+x)-x dx \\
& =y+\int_{x\in\R^d} f(A+(x+y))-(x+y) dx \\ 
& =y+f^{(1)}(A).
\end{align*}
Second, we introduce additivity:
\[f^{(2)}(A)=\int_{K\in\mathcal K} f^{(1)}(A+K)-f^{(1)}(K)dK.\]
We again note that $d(A+K,K)\leq \max_{a\in A}\|a\|_2$, so the integrand is an $L^{\infty}$ function on $\mathcal K$. As above, we now have additivity:
\begin{align*}
f^{(2)}(A+K') & =\int_{K\in\mathcal K} f^{(1)}(A+K'+K)-f^{(1)}(K) dK \\
& =\int_{K\in\mathcal K} f^{(1)}(A+(K'+K))-f^{(1)}(K'+K)+f^{(1)}(K'+K)-f^{(1)}(K) dK \\
& =f^{(2)}(A)+f^{(2)}(K')
\end{align*}
Furthermore we preserve translation invariance; in fact it is a special case of additivity where we take $K'$ to be a single point. Finally, we introduce $O(n)$ invariance, this time using an actual integral:
\[f^{(3)}(A)=\int_{g\in O(n)} g^{-1}f^{(2)}(g(A))dg.\]

This preserves additivity because the $O(n)$-action consists of linear maps, and similarly to above it results in $f^{(3)}$ additionally being $O(n)$-invariant.
Since this $\tilde f=f^{(3)}$ is isometry invariant and additive under Minkowski sum, and is also still Lipschitz since $f$ was, we conclude from Lemma~\ref{lem:steinerunique} that $\tilde f(K)=s(K)$.

Finally, we want to argue that by doing these symmetrizations, we preserve the movement upper bound $C(d,T)$. The point is that if we evaluate the movement of $\tilde f$ on a sequence of nested bodies \[K_T+B_1\supseteq K_1\supseteq \dots \supseteq K_T\] the movement is expressible as an ``average" of the movement of $f$ over a sequence of bodies which are isometries of $K_t+K'$. For any $K'$, we still have 
\[K_T+K'+B_1\supseteq K_1+K'\supseteq\dots\supseteq K_T+K'.\] Hence the sequence $(\tilde f(K_t))_{t\leq T}$ is an average of isometries applied to sequences $(f(K_t+K'))_{t\leq T}$, which can only shrink the distances. So we get that the Steiner point $\tilde f$ also has at most $C(d,T)$ movement since $f$ does.

\end{proof}

We remark that in the proof above, we did not know that $\tilde f$ was a selector until we deduced from Lemma~\ref{lem:steinerunique} that $\tilde f$ is exactly the Steiner point. Indeed, the selector property is not obviously preserved through the symmetrizations we did. We needed $f$ to be a selector in the proof so that e.g. $f(A+x)-x$ would be bounded only depending on $A$.

\section{Analysis of $\mathtt{ChasingNormedSpace}$}

\subsection{Notations and basic facts about log-concave distribution}
For any convex set $K$ and any convex function $\phi$, we use $x\sim(\phi,K)$
to denote that $x$ is sampled from the distribution $e^{-\phi(x)}1_{x\in K}/\vol_{\phi}(K)$. In particular,
we have 
\[
\E_{x\sim(\phi,K)}f(x)=\frac{\int_{K}f(x)e^{-\phi(x)}dx}{\int_{K}e^{-\phi(x)}dx}\quad\text{and}\quad\Var_{x\sim(\phi,K)}f(x)=\E_{x\sim(\phi,K)}\left(f(x)-\E_{y\sim(\phi,K)}f(y)\right)^{2}.
\]

Now, we list some facts that we will use in this section. This lemma shows that 
any log-concave distribution concentrates around any subset of constant measure.
\begin{lem}[{Borell's Lemma \cite[Lemma 2.4.5]{brazitikos2014geometry}}]
\label{lem:borell}For any log-concave distribution $p$ on $\R^{d}$, any symmetric convex set $A$ with $p(A)=\alpha\in(0,1)$, and
any $t>1$, we have
\[
1-p(tA)\leq\alpha\left(\frac{1-\alpha}{\alpha}\right)^{\frac{t+1}{2}}.
\]
\end{lem}

In particular, it concentrates around intervals centered at the center of gravity.
\begin{lem}[Tail bound of log-concave distribution]
\label{lem:tail_logconcave}For any log-concave distribution $p$
on $\R^{d}$ and any unit vector
$v$, we have that
\[
\mathbb{P}_{y\sim p}\left(\left|v^{\top}(y-x)\right|\geq t\sigma\right)=e^{-\Omega(t)}
\]
where $x=\E_{y\sim p}y$ and $\sigma=\Var_{y\sim p}v^{\top}y$.
\end{lem}

The localization lemma shows that one can prove a statement 
about high dimensional log-concave distributions via a $1$ dimensional statement.
\begin{thm}[{Localization Lemma \cite{kannan1995isoperimetric, fradelizi2004extreme}}]
Suppose $p$ is a log-concave distribution in $\R^{d}$, $g$ and
$h$ are continuous function such that 
\[
\int_{\R^{d}}g(x)p(x)dx>0\text{ and }\int_{\R^{d}}h(x)p(x)dx=0.
\]
Then, there is an interval $[a,b]\subset\R^{d}$ and scalars $\alpha,\beta$
such that
\[
\int_{[a,b]}g(x)\cdot e^{\alpha x+\beta}dx>0\text{ and }\int_{[a,b]}h(x)\cdot e^{\alpha x+\beta}dx=0.
\]
\end{thm}

This lemma relates the deviation and width of a convex set.
\begin{lem}[{\cite[Theorem 4.1]{kannan1995isoperimetric}}]
\label{lem:deviation_width}For any convex set $K\subset\R^{d}$
and any unit vector $v$, we have
\[
\mathrm{width}_{K}(v)^{2}\leq4d\cdot(d+2)\cdot\Var_{x\sim K}v^{\top}x.
\]
where $\mathrm{width}_{K}(v) = \max_{x \in K} v^\top x - \min_{x \in K} v^\top x$.
\end{lem}

Grunbaum's Theorem shows that for any log concave distribution, any
half space containing the centroid has at least $\frac{1}{e}$ fraction of the mass.
The following theorem shows that a similar statement holds as long
as the centroid is close to the half space. \cite[Theorem 3]{bertsimas2004solving}
proved this theorem for convex sets. For completeness, we include
the proof for log-concave distribution.
\begin{thm}[Grunbaum's Theorem \cite{grun}]
\label{thm:grunbaum_thm}Let $f$ be any log-concave distribution
on $\R^{d}$ and $H$ be any half space $H=\{x\in\R^{d}:v^{\top}x\geq b\}$
for some vector $v$. Let $c$ be the centroid of $f$. Then
\[
\int_{H}f(x)dx\geq\frac{1}{e}-t^{+}
\]
where $t=\frac{b-v^{\top}c}{\sqrt{\E_{x\sim f}(v^{\top}(x-c))^{2}}}$
is the distance of the centroid and the half space divided by the
deviation on the normal direction $v$.
\end{thm}

\begin{proof}
Since the marginal of a log-concave distribution is log-concave, by
taking marginals and rescaling, it suffices to prove that $\int_{x\geq t}f(x)dx\geq\frac{1}{e}-t^{+}$
for isotropic log-concave distribution on $\R$. The case for $t\leq0$
follows from the classical Grunbaum Theorem. The case for $t\geq0$
follows from the case $t=0$ and the fact that $f(x)\leq1$ for all
$x$ \cite[Lemma 5.5a]{lovasz2007geometry}.
\end{proof}
Finally, we will need this statement about strongly log-concave distribution, namely, a log-concave distribution multiplied by a Gaussian distribution.
\begin{thm}[Brascamp-Lieb inequality \cite{BL}]
\label{thm:Brascamp-Lieb}Let $\gamma$ be any Gaussian distribution on $\R^{d}$ and $f$
be any log-concave function on $\R^{d}$. Define the density function $h$ as follows:
\[
h(x)=\frac{f(x)\gamma(x)}{\int_{\R^{d}}f(y)\gamma(y)\,dy}.
\]
For any vector $v\in\R^{d}$ and any $\alpha\ge1$, $\E_{h}|v^{T}(x-\E_{h}x)|^{\alpha}\le\E_{\gamma}|v^{T}x|^{\alpha}$.
\end{thm}

\subsection{Cutting along $\protect\cg_{\phi}(\Omega)$}

The goal of this section is to prove Theorem \ref{thm:movement_cg_r}.
This theorem shows that a cutting plane passing through $\cg_{\phi}(\Omega+B_{r})$
behaves very similar to $\cg_{\phi}(\Omega)$ as long as we do not
cut along a narrow direction of $\Omega+B_{r}$. In particular, we
show that it decreases the volume $\vol_{\phi}(\Omega+B_{r})$ by a constant
factor and that $\cg_{\phi}(\Omega+B_{r})$ does not move a lot when
we cut $\Omega$. 

We will first prove $\cg_{\phi}(\Omega)$ is stable when a constant
fraction of $\Omega$ is removed.
\begin{lem}
\label{lem:movement_cg}Let $\phi$ be an $\alpha$-strongly convex
function on the normed space $\|\cdot\|$. Let $\widetilde{\Omega}\subseteq\Omega$
be two convex sets. Then, we have that
\[
\|\cg_{\phi}(\widetilde{\Omega})-\cg_{\phi}(\Omega)\|\lesssim\sqrt{\frac{\vol_{\phi}(\Omega)}{\vol_{\phi}(\widetilde{\Omega})}}\cdot\alpha^{-\frac{1}{2}}.
\]
\end{lem}

\begin{proof}
For any $h$, we have that
\begin{align*}
\Var_{x\sim(\phi,\Omega)}x^{\top}h & =\frac{\vol_{\phi}(\widetilde{\Omega})}{\vol_{\phi}(\Omega)}\frac{\int_{\Omega}((x-\cg_{\phi}(\Omega))^{\top}h)^{2}e^{-\phi(x)}dx}{\int_{\widetilde{\Omega}}e^{-\phi(x)}dx}\\
 & \geq\frac{\vol_{\phi}(\widetilde{\Omega})}{\vol_{\phi}(\Omega)}\E_{x\sim(\phi,\widetilde{\Omega})}((x-\cg_{\phi}(\Omega))^{\top}h)^{2}\\
 & \geq\frac{\vol_{\phi}(\widetilde{\Omega})}{\vol_{\phi}(\Omega)}(\E_{x\sim(\phi,\widetilde{\Omega})}(x-\cg_{\phi}(\Omega))^{\top}h)^{2}\\
 & =\frac{\vol_{\phi}(\widetilde{\Omega})}{\vol_{\phi}(\Omega)}((\cg_{\phi}(\widetilde{\Omega})-\cg_{\phi}(\Omega))^{\top}h)^{2}
\end{align*}
where the first inequality follows from $\widetilde{\Omega}\subset\Omega$,
the second inequality follows from Cauchy Schwarz. To bound the movement of $\cg$, this calculation shows that it suffices to upper bound the variance. 

We claim that
for any $\|h\|_{*}=1$, we have 
\begin{equation}
\Var_{x\sim(\phi,\Omega)}x^{\top}h\lesssim\alpha^{-1}.\label{eq:var_claim}
\end{equation}
First, assuming this claim is true, we have that $((\cg_{\phi}(\widetilde{\Omega})-\cg_{\phi}(\Omega))^{\top}h)^{2}\lesssim\frac{\vol_{\phi}(\Omega)}{\vol_{\phi}(\widetilde{\Omega})}\cdot\alpha^{-1}$.
Taking $h$ so that $(\cg_{\phi}(\widetilde{\Omega})-\cg_{\phi}(\Omega))^{\top}h=\|(\cg_{\phi}(\widetilde{\Omega})-\cg_{\phi}(\Omega)\|$,
we have that the result that $\|\cg(\widetilde{\Omega})-\cg(\Omega)\|^{2}\lesssim\frac{\vol_{\phi}(\Omega)}{\vol_{\phi}(\widetilde{\Omega})}\cdot\alpha^{-1}.$ So to prove the lemma we just need to prove this claim.

Now, we prove the claim using the localization lemma.
Consider the integral problem
\[
\max_{\text{log-concave distribution }p}\int(x^{\top}h)^{2}\cdot e^{-\phi(x)}p(x)dx\text{ subject to }\int x^{\top}h\cdot e^{-\phi(x)}p(x)dx=t.
\]
To prove our claim, it suffices to upper bound the maximum of the
integral problem by $t^{2}+O(\alpha^{-1})$ for all $t$. The localization
lemma shows that it suffices to consider the case that $p$ is log-affine. This
is same as proving 
\[
\Var_{x\sim q}(x^{\top}h)^{2}\lesssim\alpha^{-1}
\]
for any distribution $q$ of the form $q(x)\propto e^{\alpha^{\top}x-\phi(x)}$
restricted to an interval $[a,b]\subset\R^d$. We can parameterize
the distribution $q(s)$ by $s=x^{\top}h$. (Similarly for $\phi$).
Using the assumption that $\phi$ is $\alpha$-strongly convex in
$\|\cdot\|$, we have that 
\[
\phi(s)\geq\phi(t)+\left\langle \nabla\phi(t),s-t\right\rangle +\frac{\alpha}{2}\left\Vert \frac{a-b}{(a-b)^{\top}h}\right\Vert ^{2}(s-t)^{2}.
\]
Hence, $\phi(s)$ is $\frac{\alpha\|a-b\|^{2}}{((a-b)^{\top}h)^{2}}$-strongly
convex in $\R$. The Brascamp-Lieb inequality shows that 
\[
\Var_{x\sim q}(x^{\top}h)^{2}\lesssim\frac{((a-b)^{\top}h)^{2}}{\alpha\|a-b\|^{2}}.
\]
Since $\|h\|_{*}=1$, we have that $(a-b)^{\top}h\leq\|a-b\|$. Therefore,
we proved the claim (\ref{eq:var_claim}).
\end{proof}
Now, we are ready to prove the main theorem of this section:

\movementcgr*

\begin{proof}
Let $\Omega=K+B_{r}$ and that $\widetilde{\Omega}=(K\cap H)+B_{r}$.
Define the half space $\widetilde{H}=H+B_{r}$, namely, $\widetilde{H}=\{x:v^{\top}x\geq v^{\top}\cg_{\phi}(\Omega)-r\}$.
Using that $v$ is an unit vector, we have that
\[
\Omega\cap H\subseteq\widetilde{\Omega}\subseteq\Omega\cap\widetilde{H}.
\]

Now, we apply Grunbaum's Theorem and the inclusion above to prove that
$\vol_{\phi}(\widetilde{\Omega})\sim\vol_{\phi}(\Omega)$.

\textbf{Upper bound:} By the generalized Grunbaum Theorem (Theorem
\ref{thm:grunbaum_thm}), since the distance of $\widetilde{H}^{c}$
and $\cg_{\phi}(\Omega)$ is less than $\frac{1}{2e}$ times the deviation,
we have that $\vol_{\phi}(\Omega\cap\widetilde{H}^{c})\geq\frac{1}{2e}\vol_{\phi}(\Omega)$.
Hence, we have $\vol_{\phi}(\widetilde{\Omega})\leq\vol_{\phi}(\Omega\cap\widetilde{H})\leq(1-\frac{1}{2e})\cdot\vol_{\phi}(\Omega)$.

\textbf{Lower bound:} By the standard Grunbaum Theorem, since $H$
is a half space through the $\cg_{\phi}(\Omega)$, we have that $\vol_{\phi}(\Omega\cap H)\geq\frac{1}{e}\vol_{\phi}(\Omega)$.
Hence, we have $\vol_{\phi}(\widetilde{\Omega})\geq\frac{1}{e}\vol_{\phi}(\Omega)$. 

The movement of $\cg_{\phi}$ follows from this volume lower bound
and Lemma \ref{lem:movement_cg}.
\end{proof}

\subsection{Main Result}

Now, we are ready to upper bound the movement per step.
\begin{lem}
\label{lem:movement}Each step $x$ moves by at most $O(C\cdot dD\cdot r+\alpha^{-\frac{1}{2}})$.
\end{lem}

\begin{proof}
\textbf{Movement of $\cg$:}

The half space $H$ satisfies $x+V\subset\partial H$. Hence, this
half space passes through $x=\cg_{\phi}(\Omega+B_{r})$ and has
normal vector $v\in V^{\perp}$. From the definition of $V$, we have
that $\Var_{y\sim(\phi,\Omega+B_{r})}v^{\top}y\geq(2e\cdot r)^{2}$.
This satisfies the assumption in Theorem \ref{thm:movement_cg_r}
and hence, we have that the movement of $\cg_{\phi}$ is bounded by
$O(\alpha^{-\frac{1}{2}})$.

\textbf{Movement of the nested subroutine:}

For any $v\in V$, we have that $\Var_{y\sim(\phi,\Omega+B_{r})}v^{\top}y<(2e\cdot r)^{2}$
from the definition of $V$. Lemma \ref{lem:tail_logconcave} shows
that only $\frac{1}{2}e^{-D}$ fraction of the distribution $e^{-\phi(y)}1_{y\in\Omega+B_{r}}$
lies outside the interval $I:=\left\{ \left|v^{\top}(y-x)\right|\leq c\cdot(1+D)\cdot r\right\} $
for some universal constant $c$. Since $0\leq\phi(x)\leq D$ for
all $x\in\Omega+B_{r}$, this implies that at least half of the distribution
$1_{y\in\Omega+B_{r}}$ lies in this interval $I$.

Now, Borell's Lemma (Lemma \ref{lem:borell}) shows that the distribution $1_{y\in\Omega+B_{r}}$ concentrates around this interval and hence the deviation
of $\Omega+B_{r}$ on direction $v$ is bounded by $O((1+D)\cdot r)=O(D\cdot r)$.
(It is elementary to see $D\geq\frac{1}{2}$ always holds.) Now, Lemma \ref{lem:deviation_width}
shows that $\mathrm{width}_{\Omega+B_{r}}(v)=O(dD\cdot r)$ for all
$v\in V$. Since $K\subset\Omega$, $(x+V)\cap K$ has diameter $O(dD\cdot r)$.
Therefore, the $C$-competitive algorithm can induce at most $O(C\cdot dD\cdot r)$
movement.

Combining both parts, the movement per step is at most $O(C\cdot dD\cdot r+\alpha^{-\frac{1}{2}})$.
\end{proof}
\begin{lem}
\label{lem:iteration}There are at most $O(D+d\log\frac{d}{r})$ iterations.
\end{lem}

\begin{proof}
We measure the progress by $\vol_{\phi}(\Omega+B_{r})$. As we discussed
in the proof of Lemma \ref{lem:movement}, the assumption in Theorem
\ref{thm:movement_cg_r} is satisfied. Hence, $\vol_{\phi}(\Omega+B_{r})$
is decreased by $1-\frac{1}{2e}$ factor every step. 

Since $\phi\geq0$, the initial volume is upper bounded by
\[
\vol_{\phi}(\{\|x\|\leq1+\sqrt{d}\cdot r\})\leq\vol(\{\|x\|\leq2\})\leq\vol(\{\|x\|_{2}\leq2\sqrt{d}\}).
\]
Next, we note that $\phi\leq D$ and hence
\[
\vol_{\phi}(\Omega+B_{r})\geq\vol_{\phi}(B_{r})\geq e^{-D}\cdot\vol(\{\|x\|\leq r\})\geq e^{-D}\cdot\vol(\{\|x\|_{2}\leq r\}).
\]
Therefore, the $\vol_{\phi}(\Omega+B_{r})$ can at most decreased
by a factor of $e^{D}(\frac{2\sqrt{d}}{r})^{d}$ during the whole
algorithm. Since each iteration the volume is decreased by a constant factor,
the algorithm ends in $O(D+d\log\frac{d}{r})$ steps.
\end{proof}

\normedchasing*

\begin{proof}
Lemma \ref{lem:reduction} shows that it suffices to give an online algorithm with total movement $O(\sqrt{d D \cdot \log{d})}$ for the case the body is bounded in a unit norm ball.

As we proved using Steiner point, there is a $O(d)$-competitive
algorithm for $\ell^{2}$ nested chasing convex bodies. Hence, we can
take $C=O(d)$ in the algorithm $\mathtt{ChasingNormedSpace}$. Also,
we can do a change of variable for the normed space using the John ellipsoid to satisfy the assumption $\|x\|_{2}\leq\|x\|\leq\sqrt{d}\cdot\|x\|_{2}$.

Lemma \ref{lem:movement} and Lemma \ref{lem:iteration} show that
the algorithm $\mathtt{ChasingNormedSpace}$ is
\[
O(1)\cdot(D_{\alpha}+d\log\frac{d}{r})\cdot(d^{2}D_{\alpha}\cdot r+\alpha^{-\frac{1}{2}})
\]
competitive if there is an $\alpha$-strongly convex $\phi_{\alpha}$
on the normed space with $0\leq\phi_{\alpha}(x)\leq D_{\alpha}$ for
all $\|x\|\leq2$. 

Using the assumption of this theorem, we can rescale $\phi$ and
the domain to construct $\phi_{\alpha}$ such that it is $\alpha$-strongly convex $\phi$ with $D_{\alpha}=O(\alpha\cdot D)$. Hence,
we have the competitive ratio
\[
O(1)\cdot(\alpha D+d\log\frac{d}{r})\cdot(d^{2}\alpha D\cdot r+\alpha^{-\frac{1}{2}}).
\]
Taking $\alpha=\frac{d}{D}(1+\log\frac{1}{r})$ and $r=\frac{1}{d^{2}\alpha^{\frac{3}{2}}D}$,
we have the competitive ratio
\[
O(1)\cdot\sqrt{dD\cdot(1+\log d^{2}\alpha^{\frac{3}{2}}D)}.
\]
Finally, we note that $D\geq\frac{1}{2}$ (using the definition of
strong convexity) and $D\leq\frac{d}{2}$ (otherwise, we can use $\|x\|^{2}$ as $\phi$).
Therefore, $r\leq\frac{1}{\sqrt{d}}$ and that the competitive ratio
is simply $O(\sqrt{dD\cdot\log d})$.
\end{proof}

\end{document}